\documentclass[reqno,12pt]{amsart}
\usepackage{amsmath,amssymb,latexsym,amscd,amsthm}

\usepackage[english]{babel}
\usepackage{times}
\usepackage[T1]{fontenc}

\def\calD{{\mathcal D}}

\def\calI{{\mathcal I}}
\def\calJ{{\mathcal J}}
\def\calP{{\mathcal P}}

\def\calY{{\mathcal Y}}

\def\pol#1{\langle #1 \rangle}

\def\ncp#1#2{#1\langle #2\rangle}
\def\ncs#1#2{#1\langle\langle #2\rangle\rangle}

\def\KY{\K \langle Y \rangle}

\def\scal#1#2{\langle #1\bv#2 \rangle}
\def\bv{\mid}

\def\Haus{\operatorname{\mathcal H\text{$aus$}}}
\def\Iplus{\operatorname{I_+}}
\def\conc{{\tt conc}}

\gdef\stufflemin{\;%
  \setlength{\unitlength}{0.0125cm}%
  \begin{picture}(20,10)(220,580)
  \thinlines
  \put(220,592){\line( 0,-1){ 10}}
  \put(220,582){\line( 1, 0){ 20}}
  \put(240,582){\line( 0, 1){ 10}}
  \put(225,587){\line( 1, 0){ 10}}
  \end{picture}\;
}

\gdef\svuffle{\;%
  \setlength{\unitlength}{0.0125cm}%
  \begin{picture}(20,10)(220,580)
  \thinlines
  \put(220,592){\line( 0,-1){ 10}}
  \put(220,582){\line( 1, 0){ 20}}
  \put(240,582){\line( 0, 1){ 10}}
  \put(230,587){\circle{5}}
  \end{picture}\;
}

\newcommand{\calA}{{\mathcal A}}

\newcommand{\calH}{{\mathcal H}}
\newcommand{\calC}{{\mathcal C}}

\newcommand{\calS}{{\mathcal S}}

\newcommand{\N}{{\mathbb N}}

\newcommand{\Q}{{\mathbb Q}}

\newcommand{\K}{{\mathbb K}}

\newcommand{\Frac}[2]{\displaystyle \frac{#1}{#2}}
\newcommand{\Sum}[2]{\displaystyle{\sum_{#1}^{#2}}}

\def\pol#1{\langle #1 \rangle}
\def\Lyn{\operatorname{\mathcal L\text{$yn$}}}

 \def\shuffle{\mathop{_{^{\sqcup\!\sqcup}}}} 

\gdef\stuffle{\;%
  \setlength{\unitlength}{0.0125cm}%
  \begin{picture}(20,10)(220,580) 
  \thinlines 
  \put(220,592){\line( 0,-1){ 10}} 
  \put(220,582){\line( 1, 0){ 20}} 
  \put(240,582){\line( 0, 1){ 10}} 
  \put(230,592){\line( 0,-1){ 10}} 
  \put(225,587){\line( 1, 0){ 10}} 
  \end{picture}\; 
}

\newtheorem{corollary}{Corollary}
\newtheorem{proposition}{Proposition}
\newtheorem{theorem}{Theorem}
\newtheorem{lemma}{Lemma}
\newtheorem{definition}{Definition}

\theoremstyle{remark}
\newtheorem{remark}{Remark}
\newtheorem{example}{Example}

\newcommand{\serie}[2]{#1 \langle \! \langle #2 \rangle \! \rangle}

\def\pol#1{\langle #1 \rangle}

\def\KY{\K \langle Y \rangle}

\newcommand{\calB}{{\mathcal B}}

\def\Sum{\displaystyle\sum}

\def\Frac{\displaystyle\frac}
\def\path{\rightsquigarrow}

\def\bv{\mid}

\gdef\minishuffle{{\scriptstyle \shuffle}}  
\gdef\ministuffle{{\scriptstyle \stuffle}}

\def\supp{\mathop\mathrm{supp}\nolimits}

\begin{document}

\newbox\Adr
\setbox\Adr\vbox{
\centerline{$^\dagger$Universit\'e Paris XIII, 1, 93430 Villetaneuse, France}
\centerline{$^\ddagger$LIPN - UMR 7030, CNRS, 93430 Villetaneuse, France}
\centerline{$^*$Universit\'e Lille II, 1, Place D\'eliot, 59024 Lille, France}
}

\newbox\Foot
\setbox\Foot\vbox{\footnotesize\noindent$^1$ 
The present work is part of a series of papers devoted to the study of the renormalization
of divergent polyzetas (at positive and at non-positive indices) via the factorization
of the non-commutative generating series of polylogarithms and of harmonic sums,
and via the effective construction of pairs of dual bases in duality in $\varphi$-deformed shuffle algebras.
It is a sequel to \cite{BDHMT}, and its content was presented in several seminars and meetings, including the 74th S\'eminaire Lotharingien de Combinatoire.}

\title[(Pure) transcendence bases in $\varphi$-deformed shuffle bialgebras]
{(Pure) transcendence bases in $\varphi$-deformed shuffle bialgebras$^1$\footnote{\kern-13pt\box\Foot}}

\author[V.C. Bui, G.H.E. Duchamp, Q.H. Ng\^o, V. Hoang Ngoc Minh, and C. Tollu]
{V.C. Bui$^{\ddagger}$, G.H.E. Duchamp$^{\dagger\,\ddagger}$, 
Q.H. Ng\^o$^{\ddagger}$, V. Hoang Ngoc Minh$^{*\,\ddagger}$, 
and C.~Tollu$^{\dagger\,\ddagger}$\\\\[50pt]\box\Adr}

\begin{abstract}
Computations with integro-differential operators are often carried out in an associative algebra with unit,
and they are essentially non-commutative computations. By adjoining a cocommutative co-product,
one can have those operators 
act on a bialgebra isomorphic to an enveloping algebra. 
That gives an adequate framework
for a computer-algebra implementation via monoidal factorization, (pure) transcendence bases and Poincar\'e--Birkhoff--Witt bases.

In this paper, we systematically study these deformations, obtaining necessary and sufficient conditions for the operators to exist, and we give the most general cocommutative deformations of the shuffle co-product
and an effective construction of pairs of bases in duality. The paper ends by the combinatorial setting of local systems of coordinates on the group of group-like series.
\end{abstract}

\maketitle

\addtocounter{footnote}{1}

\tableofcontents

\thispagestyle{myheadings}
\font\rms=cmr8 
\font\its=cmti8 
\font\bfs=cmbx8

\markright{\its S\'eminaire Lotharingien de
Combinatoire \bfs 74 \rms (2018), Article~B74f\hfill}
\def\thepage{}

\section{Introduction}

The shuffle product first appeared in 1953 in the work of Eilenberg and Mac Lane \cite{EML}. As soon as 1954, Chen used it to express the product of iterated (path) integrals \cite{Chen1954}, and Ree, building on Friedrichs' criterion, proved that the non-commutative generating series of iterated integrals are exponentials of Lie polynomials, thus connecting the Lie polynomials with the shuffle product \cite{Ree}. In 1956, Radford proved that the Lyndon words form a (pure) transcendence basis of the shuffle algebra \cite{radford}. The latter result is now well understood through the duality between bialgebras and enveloping algebras (see for example \cite{reutenauer}), of which the construction in 1958 of the Poincar\'e--Birkhoff--Witt\footnote{From now on, Poincar\'e--Birkhoff--Witt will be abbreviated to PBW.}--Lyndon basis by Chen, Fox and Lyndon \cite{lyndon} and of its dual basis by Sch\"ut\-zenberger, via monoidal factorization \cite{Shutz1958,reutenauer}, gave a striking illustration. This pair of dual bases enabled 
one to factorize the diagonal series in the shuffle bialgebra and, consequently, to proceed combinatorially with the Dyson series \cite{FPSAC96} or the transport operator \cite{hoangjacoboussous}, which play a leading role in the relations between special functions involved in the theory of quantum groups \cite{kassel} and in number theory \cite{cartier}.

In 1973, that is, within twenty years of the introduction of the shuffle product, Knutson defined the quasi-shuffle in \cite{knutson}, where it 
shows up as the inner product of functions on the symmetric groups\footnote{In the present paper, that product will be referred to as the quasi-shuffle or as the stuffle product, indifferently.}. This product is very similar to the Rota--Baxter operator introduced by Cartier in 1972, in his study of the so-called Baxter algebras \cite{cartier72}. Although the analogue of Radford's theorem was pointed out by Malvenuto and Reutenauer \cite{MalvenutoReutenauer},
the factorization of the diagonal series in the quasi-shuffle bialgebra, initiated in \cite{acta,VJM}, has not yet been carried over to more general bialgebras.

\pagenumbering{arabic}
\addtocounter{page}{1}
\markboth{\SMALL V.C. BUI, G.H.E. DUCHAMP, Q.H. NG\^O, V. HOANG NGOC MINH, AND C. TOLLU}{\SMALL (PURE) TRANSCENDENCE BASES IN $\varphi$-DEFORMED SHUFFLE
BIALGEBRAS}

Sch\"utzenberger's factorization\footnote{Also called MSR factorization after the names of M\'elan\c{c}on, 
Sch\"utzenberger and Reutenauer.} \cite{reutenauer} and its extensions have since been applied to the renormalization of the associators \cite{acta,VJM}, 
to which matter they turned out to be central\footnote{These associators, which are formal power series in non-commutative variables, were introduced in quantum field theory by Drinfel'd \cite{drinfeld2}. The explicit coefficients of the universal associator $\Phi_{KZ}$ are polyzetas and regularized polyzetas \cite{lemurakami}.}. 

The coefficients of these power series are polynomial functions of positive integral multi-indices of Riemann's zeta function\footnote{These values are usually referred to as MZV's by Zagier \cite{zagier} and as polyzetas by Cartier \cite{cartier}.}  \cite{lemurakami,zagier}, and they satisfy quadratic relations \cite{cartier} which Lyndon words help explicit and explain.
The latter relations 
can be obtained by identifying the local coordinates on a bridge equation connecting the Cauchy and the Hadamard algebras of polylogarithmic functions, and by using the factorization of the non-commutative generating series of polylogarithms \cite{FPSAC99} and of harmonic sums \cite{acta,VJM}. This bridge equation is mainly a consequence of the isomorphisms between the algebra of non-commutative generating series of polylogarithms and the shuffle algebra on 
the one hand,
and between the algebra of non-commutative generating series of harmonic sums and the quasi-shuffle algebra on the other hand.

As for the generalization of Sch\"utzenberger's factorization to more general bialgebras, the key step, and the main difficulty thereof, is to decompose 
orthogonally such bialgebras into the Lie algebra generated by its primitive elements and the associated orthogonal ideal, as Ree was able to achieve in the case of the shuffle bialgebra \cite{Ree}, and to construct, whenever possible, the respective bases. In favorable cases, 
it is to be hoped
that those bialgebras are enveloping algebras, so that the Eulerian projectors are convergent and other analytic computations can be performed.

To make that decomposition possible, one first needs to determine the Eulerian projectors by taking the logarithm of the diagonal series
and second to insure their convergence. A key ingredient is the fact that the diagonal series are group-like and give a host
of group-like elements by specialization, so one can use the exponential-logarithm correspondence to compute within a combinatorial Hausdorff group.

To that effect, the present work generalizes the recursive definitions of the shuffle and quasi-shuffle products given
by Fliess \cite{fliess} and Hoffman \cite{hoffman}, 
respectively, to introduce the $\varphi$-deformed shuffle product,
where $\varphi$ stands for an arbitrary algebra law. Recent studies on these structures can be found in \cite{Dutopeko,Manchon,Patras}.
 
These $\varphi$-shuffle products interpolate between the classical shuffle and quasi-shuffle products (for $\varphi\equiv0$ and $\varphi\equiv1$, respectively), and allow a classification of the associated bialgebras.

This paper is devoted to the combinatorics of $\varphi$-deformed shuffle algebras and to 
the effective constructions of pairs of dual bases. 
Its organisation is as follows:
\begin{itemize}
\item Section~\ref{primer} is a short reminder of well-known facts about the combinatorics of the $q$-stuffle product \cite{BDM},
which encompasses the shuffle \cite{reutenauer} and the quasi-shuffle products \cite{acta,VJM}.
\item In Section~\ref{algebra}, we thoroughly investigate algebraic and combinatorial aspects of the $\varphi$-deformed shuffle products and explain how to use bases in duality to get a local system of coordinates on the (infinite-dimensional) Lie group of group-like series. 
\end{itemize}

Throughout the paper, we have a particular concern for Lie series and their correspondence with the Hausdorff group.

\section{A survey of shuffle products}\label{primer}
For standard definitions and facts pertaining to the (algebraic) combinatorics on words, we refer the reader to the classical books by Lothaire \cite{lothaire} and Reutenauer \cite{reutenauer}.

Throughout the paper, $\K$ stands for a (unital, associative and commutative) $\Q$-algebra containing a parameter $q$. In this section, we review the known combinatorics of bases in duality and local coordinates on the infinite-dimensional Lie group of group-like series (Hausdorff group). The parameter $q$ allows for a unified treatment between shuffle ($q=0$) and quasi-shuffle ($q=1$) products. 

Let $Y=\{y_i\}_{i\ge1}$ be an alphabet, totally ordered by $y_1>y_2>\cdots$.
The free monoid and the set of Lyndon words over $Y$ are denoted by $Y^*$ and $\Lyn Y$, respectively.
The unit of $Y^*$ is denoted by $1_{Y^*}$. We also write $Y^+=Y^*\setminus\{1_{Y^*}\}$.

The $q$-stuffle \cite{BDM}, which interpolates between the shuffle \cite{Ree}, quasi-shuffle \cite{MalvenutoReutenauer}
(or stuffle) and minus-stuffle products \cite{JSC}, for 
$q=0$, $1$, and $-1$, respectively, is defined as follows:
\begin{align}
   u\stuffle_q 1_{Y^*}&=1_{Y^*}\stuffle_q u=u,\\
   y_su\stuffle_q y_tv&=y_s(u\stuffle_q y_tv)+y_t(y_su\stuffle_q v) +q\,y_{s+t}(u\stuffle_q v),
\end{align}
or its dual co-product, as follows, for any $y_s,y_t\in Y$ and $u,v\in Y^*$,
 \begin{align}
   \Delta_{\stuffle_q}(1_{Y^*})&=1_{Y^*}\otimes1_{Y^*},\\
   \Delta_{\stuffle_q}(y_s)&=y_s\otimes 1_{Y^*} +1_{Y^*}\otimes y_s +q\sum_{s_1+s_2=s}{y_{s_1}\otimes y_{s_2}}.
 \end{align}

We now turn to the study of the combinatorial $q$-stuffle Hopf algebra, which we do by stressing the importance of the Lie elements\footnote{Following Ree \cite{Ree},
the Lie elements contain the non-commutative power series which are Lie series (as the Chen non-commutative
generating series of iterated integrals), i.e., they are group-like for the co-product of the shuffle.} studied by Ree \cite{Ree}, and show how Sch\"utzenberger's factorization extends to this new structure.

The $q$-stuffle is commutative, associative and unital. With the co-unit defined by $\epsilon(P)=\scal{P}{1_{Y^*}}$, for $P\in\KY$,
we get 
$$\calH_{{\stuffle_q}}=(\KY,{\tt
  conc},1_{Y^*},\Delta_{{\stuffle_q}},\epsilon)$$ 
and 
$$\calH_{{\stuffle_q}}^{\vee}=(\KY,{\stuffle_q},1_{Y^*},\Delta_{\tt
  conc},\epsilon)$$ 
which are mutually dual bialgebras and, in fact, Hopf algebras because they are $\N$-graded by the weight.

Let $\calD_Y$ be the diagonal series over $\calH_{{\stuffle_q}}$, i.e.,
\begin{equation}
\calD_Y=\sum_{w\in Y^*}w\otimes w.
\end{equation}
Then\footnote{The diagonal series lives in $\ncs{\K}{Y^*\otimes Y^*}\simeq (\ncp{\K}{Y}\otimes\ncp{\K}{Y})^*$.}
\begin{equation}
\log\calD_Y=\sum_{w\in Y^+}w\otimes\pi_1(w),
\end{equation}
where $\pi_1$ is the extended Eulerian projector\footnote{In fact,
  $\pi_1$ is a projector which maps $\calH_{{\stuffle_q}}$ onto the
  space of its primitive elements $Prim(\calH_{{\stuffle_q}})$, see
  Lemma~\ref{primitive_proj}.} over $\calH_{{\stuffle_q}}$, defined by
(see \cite{BDM})
\begin{equation}
\pi_1(w)=w+\sum_{k\ge2}\frac{(-1)^{k-1}}{k}
\sum_{u_1,\ldots,u_k\in Y^+}\langle w\bv u_1{\stuffle_q}\cdots{\stuffle_q} u_k\rangle u_1\ldots u_k.
\end{equation}

Let $\{\Pi_l\}_{l\in\Lyn Y}$ be defined by
\begin{equation}\label{pi_l0}
\begin{cases}
\Pi_y=\pi_1(y),&\mbox{for }y\in Y,\\
\Pi_{l}=[\Pi_s,\Pi_r],&\mbox{for the standard factorization $(s,r)$ of }l\in \Lyn Y-Y.
\end{cases}
\end{equation}
Then it forms a basis of the Lie algebra of primitive elements of
$\calH_{{\stuffle_q}}$ (see \cite{BDM}).

Let $\{\Pi_w\}_{w\in Y^*}$ be defined, for any $w\in Y^*$ such that $w=l_1^{i_1}\ldots l_k^{i_k}$
with $l_1>\ldots>l_k$ and $l_1\ldots,l_k\in\Lyn Y$, by
\begin{equation}
\Pi_{w}=\Pi_{l_1}^{i_1}\ldots \Pi_{l_k}^{i_k}.
\end{equation}
Then, by the PBW theorem, the set $\{\Pi_w\}_{w\in Y^*}$ is a basis of
$\KY$ (see \cite{BDM}).

Let $\{\Sigma_w\}_{w\in Y^*}$ be the family dual\footnote{The duality pairing is given by $\scal{u}{v}=\delta_{u,v}$, for $u,v\in Y^*$.} to $\{\Pi_w\}_{w\in Y^*}$ in the quasi-shuffle algebra.
Then\break $\{\Sigma_w\}_{w\in Y^*}$ freely generates the
quasi-shuffle algebra, and the subset $\{\Sigma_l\}_{l\in\Lyn Y}$
forms a transcendence basis of 
$(\KY,\stuffle_q,1_{Y^*})$. The $\Sigma_w$ can be obtained as follows
(see \cite{BDM}):
\begin{equation}\label{recurY-bis}
\begin{cases}
\Sigma_y=y,&\mbox{for }y\in Y,\\
\Sigma_l\;=\displaystyle
\sum_{(!)}\frac{q^{i-1}}{i!}y_{s_{k_1}+\cdots+s_{k_i}}\Sigma_{l_1\cdots l_n},&\mbox{for }l=y_{s_1}\ldots y_{s_k}\in\Lyn Y,\\
\Sigma_w=\displaystyle\frac{\Sigma_{l_1}^{\stuffle_q i_1}\stuffle_q\cdots\stuffle_q\Sigma_{l_k}^{\stuffle_q i_k}}{i_1!\cdots i_k!},
&\mbox{for }w=l_1^{i_1}\ldots l_k^{i_k},
\end{cases}
\end{equation}
and $l_1\succ_{lex}\cdots\succ_{lex} l_k \in \Lyn Y$. In the second 
expression 
of (\ref{recurY-bis}),
the sum  $(!)$ is taken over all subsequences $\{k_1,\ldots,k_i\}\allowbreak \subset \allowbreak\{1,\ldots,k\}$
and all Lyndon words $l_1\succeq_{lex}\cdots\succeq_{lex} l_n$ such that
$(y_{s_1},\ldots,\allowbreak y_{s_k})\stackrel{*}{\Leftarrow}(y_{s_{k_1}},\ldots,y_{s_{k_i}},l_1,\ldots,l_n)$,
where $\stackrel{*}{\Leftarrow}$ denotes the transitive closure of the relation on standard  sequences,
denoted by $\Leftarrow$  (see \cite{BDM}).

In this case, since $\{\Pi_w\}_{w\in Y^*}$ and $\{\Sigma_w\}_{w\in
  Y^*}$ are multiplicative, we get the $q$-extended 
Sch\"utzenberger's factorization as follows (see \cite{BDM}):
\begin{equation}
\calD_Y=\sum_{w\in Y^*}\Sigma_w\otimes\Pi_w=\prod_{l\in\Lyn Y}^{\searrow}\exp(\Sigma_l\otimes\Pi_l).
\end{equation}
This series, in the factorized form, encompasses a large part of the combinatorics of Dyson's functional expansions in quantum field theory
\cite{dyson,magnus}. It is the infinite-dimensional analogue of the theorem of Wei and Norman \cite{Bui,weinorman1,weinorman2}.

\section{Algebraic aspects of $\varphi$-shuffle bialgebras}\label{algebra}
From now on, we will work with an alphabet $Y=\{y_i\}_{i\in I}$ with $I$ an arbitrary index set\footnote{The indexing is one-to-one, i.e., there is no repetition.}, which needs not be totally ordered unless we write it explicitly.
\subsection{First properties}
Let us consider the following recursion in order to construct a map 
\begin{equation}
Y^*\times Y^*\longrightarrow \KY.
\end{equation}
\begin{itemize}
\item[i)] For any $w\in Y^*$,
\begin{eqnarray}
\hspace{-2cm} (\mathrm{Init}) && \ 1_{Y^*}\ministuffle_{\varphi}w=w\ministuffle_{\varphi}1_{Y^*}=w.
\end{eqnarray}
\item[ii)] For any $a, b\in Y$ and $u,v\in Y^*$, 
\begin{eqnarray}
(\mathrm{Rec})&&au\ministuffle_{\varphi}bv 
=a(u\ministuffle_{\varphi}bv)+b(au\ministuffle_{\varphi}v)+\varphi(a,b)(u\ministuffle_{\varphi} v),
\end{eqnarray}
where $\varphi$ is an arbitrary mapping defined by its structure coefficients
\begin{eqnarray}\label{phidef}
\varphi:Y\times Y&\longrightarrow&\K Y,\\
(y_i,y_j)&\longmapsto&\sum_{k\in I}\gamma_{y_i,y_j}^{y_k}\;y_k.
\end{eqnarray}
\end{itemize}
The following proposition guarantees the existence of a unique bilinear law on $\ncp{\K}{Y}$ satisfying $(\mathrm{Init})$ and $(\mathrm{Rec})$. 
\begin{proposition}[\cite{BDHMT}]
The recursion $(\mathrm{Rec})$ together with the initialization $(\mathrm{Init})$ defines a unique mapping
\begin{eqnarray*}
\stuffle_\varphi\,: Y^*\times Y^*&\longrightarrow&\KY,
\end{eqnarray*}
which can, at once, be extended by linearity as a law
\begin{eqnarray*}
\stuffle_\varphi\,: \KY\otimes\KY&\longrightarrow&\KY.
\end{eqnarray*}
\end{proposition}

The space $\ncp{\K}{Y}$ endowed with the law $\stuffle_\varphi$ is an
algebra (with unit $1_{\ncp{\K}{Y}}$ by definition). It  will be
called the $\varphi$-shuffle algebra. In full generality, this algebra
need not be associative or commutative if $\varphi$ is not so. In the
next example, we give a table of well known laws which can be defined 
according to this pattern (in which $\varphi$ is reasonable).    

\begin{example}\label{tab1} Below, a summary table of
  $\varphi$-deformed cases found in the literature
is given. The last case
  (infiltration product) comes from computer science (see
  \cite{infiltr1,infiltr2,direct_dual}) 
$$\begin{tabular}{|c|c|c|}
\hline
\hskip-2mm Name\hskip-6mm&\hskip-2mm (recursion) Formula\hskip-6mm&\hskip-2mm   $\varphi$\cr
\hline
\hline
\hskip-2mm Shuffle\hskip-6mm&\hskip-2mm   $au\shuffle bv=a(u\shuffle bv)+b(au\shuffle v)$\hskip-6mm&\hskip-2mm   $\varphi\equiv 0$\cr
\hline
\hskip-2mm Quasi-shuffle\hskip-6mm&\hskip-2mm   $x_iu\stuffle x_jv=x_i(u\stuffle x_jv)+x_j(x_iu\stuffle v)$\hskip-6mm&\hskip-2mm   $\varphi(x_i,x_j)=x_{i+j}$\cr
Stuffle&$\phantom{x_iu\stuffle x_jv=}+x_{i+j}(u\stuffle v)$\hskip-3mm&\hskip-1mm\cr
\hline
\hskip-2mm Min-shuffle\hskip-6mm&\hskip-2mm   $x_iu\stufflemin x_jv=x_i(u\stufflemin x_jv)+x_j(x_iu\stufflemin v)$\hskip-6mm&\hskip-2mm   $\varphi(x_i,x_j)=- x_{i+j}$\cr
&$\phantom{x_iu\stufflemin x_jv=} - x_{i+j}(u\stufflemin v)$\hskip-3mm&\hskip-1mm\cr
\hline
\hskip-2mm Muffle\hskip-6mm&\hskip-1mm $x_iu\svuffle x_jv=x_i(u\svuffle x_jv)+x_j(x_iu\svuffle v)$\hskip-6mm&\hskip-1mm $\varphi(x_i,x_j)=x_{i\times j}$\cr
&$\phantom{x_iu\svuffle x_jv=}+x_{i\times j}(u\svuffle v)$\hskip-3mm&\hskip-1mm\cr
\hline
\hskip-2mm $q$-stuffle\hskip-6mm&\hskip-2mm   $x_iu\stuffle_q x_jv=x_i(u\stuffle_q x_jv)+x_j(x_iu\stuffle_q v)$\!\!&\!\! $\varphi(x_i,x_j)=qx_{i+j}$\cr
&$\phantom{x_iu\shuffle x_jv=x}+q x_{i+j}(u\shuffle v)$\hskip-3mm&\hskip-1mm\cr
\hline
\hskip-2mm $q$-stuffle\hskip-6mm&\hskip-2mm   $x_iu\shuffle_q x_jv=x_i(u\shuffle_q x_jv)+x_j(x_iu\shuffle_q v)$\hskip-6mm&\hskip-2mm   $\varphi(x_i,x_j)=q^{i\times j}x_{i+j}$\cr
(character)&$\phantom{x_iu\shuffle x_jv=}+q^{i\times j}x_{i+j}(u\shuffle v)$\hskip-3mm&\hskip-1mm\cr
\hline
\hskip-1mm ${\tiny\tt LDIAG}(1,q_s)$\hskip-3mm&\hskip-1mm\!\!&\!\!\cr
\hskip-1mm non-crossed,\hskip-6mm&\hskip-2mm   $au\ast bv=a(u\ast bv)+b(au\ast v)$\hskip-6mm&\hskip-2mm   $\varphi(a,b)=q_s^{|a||b|}(a.b)$\cr
\hskip-1mm non-shifted&$\phantom{au\shuffle bv=}+q_s^{|a||b|}(a.b) (u\ast v)$\hskip-3mm&\hskip-1mm
\hskip-3mm\hskip-1mm$\phantom{\varphi(a,b)}(a.b)$ assoc.\cr
\hline
\hskip-2mm B-shuffle\hskip-6mm&\hskip-2mm   $au\shuffle_B bv=a(u\shuffle_B bv)+b(au\shuffle_B v)$\hskip-6mm&\hskip-2mm   $\varphi(a,b)=\langle a,b\rangle$\cr
\hskip-6mm&\hskip-2mm   $\phantom{au\shuffle bv=}+\langle a,b\rangle(u\shuffle_B v)$
\hskip-3mm&\hskip-1mm$\phantom{\varphi(a,b)}=\langle b,a\rangle$\cr
\hline
\hskip-2mm Semigroup-\hskip-6mm&\hskip-2mm   $x_tu\shuffle_{\perp} x_sv=x_t(u\shuffle_{\perp} x_sv)+x_s(x_tu\shuffle_{\perp} v)$\hskip-6mm&\hskip-2mm   $\varphi(x_t,x_s)=x_{t\perp s}$\cr
\hskip-2mm -shuffle&$\phantom{x_tu\shuffle x_sv=}+x_{t\perp s}(u\shuffle_{\perp} v)$\hskip-3mm&\hskip-1mm\cr
\hline 
\hskip-2mm $q$-Infiltration\hskip-6mm&\hskip-2mm   $au\uparrow bv=a(u\uparrow bv)+b(au\uparrow v)$\hskip-6mm&\hskip-2mm   $\varphi(a,b)=q\delta_{a,b}a$\cr
\hskip-6mm&\hskip-2mm   $\phantom{au\shuffle bv=}+q\delta_{a,b}a(u\uparrow v)$
\hskip-3mm&\hskip-1mm$\phantom{\varphi(a,b)}$\cr
\hline
\end{tabular}$$
\end{example}

Now we set forth 
the first properties of $\ministuffle_{\varphi}$ (see \cite{orsay}): {\it associativity, commutativity} and {\it dualizability}. 
\begin{definition}[\cite{BDHMT}]
A law $\mu : \ncp{\K}{Y}\otimes \ncp{\K}{Y}\rightarrow \ncp{\K}{Y}$ is 
said to be {\rm dualizable} if there exists a (linear) mapping
\begin{equation*}
\Delta_{\mu}:\ncp{\K}{Y}\longrightarrow \ncp{\K}{Y}\otimes\ncp{\K}{Y}
\end{equation*}
(necessarily unique) such that the dual mapping  
\begin{equation*}
\Big(\ncp{\K}{Y}\otimes \ncp{\K}{Y}\Big)^*\longrightarrow \ncs{\K}{Y}
\end{equation*}
restricts \footnote{through the pairings $\scal{-}{-}$.} to $\mu$. 
Or, equivalently,
\begin{eqnarray*}
\text{for all } u,v,w\in Y^*:&&\scal{\mu(u\otimes v)}{w}=\scal{u\otimes v}{\Delta_{\mu}(w)}^{\otimes2}.
\end{eqnarray*}
\end{definition}
\begin{theorem}[\cite{BDHMT}]\label{delta_phi}
We have:
\begin{enumerate}
\item The law $\ministuffle_{\varphi}$ is associative (respectively commutative) if and only if the extension
$\varphi : \K Y\otimes \K Y\longrightarrow \K Y$ is so.
\item Let $\gamma_{x,y}^z:=\scal{\varphi(x,y)}{z}$ be the structure constants of $\varphi$.
Then $\ministuffle_{\varphi}$ is dualizable
if and only if $\varphi$ is also dualizable, that is to say, there
exists a map $\delta : \K Y\longrightarrow \K Y\otimes \K Y$ such that
for all $x,y,z\in X$ we have
$$
\scal{\varphi(x,y)}{z}=\scal{x\otimes y}{\delta(z)}\ .
$$
This map $\delta$ is given by\footnote{Note that all these conditions are equivalent to the fact that $(\gamma_{x,y}^z)_{x,y,z\in Y}$ satisfies:
$$
\text{for all }z\in Y:\ \#\{(x,y)\in Y^2|\gamma_{x,y}^z\not=0\}<+\infty\ .
$$} 
$$
\delta(z)=\sum_{x,y\in Y}\gamma_{x,y}^z\,x\otimes y\ .
$$
\end{enumerate}
\end{theorem}


For the proof of the theorem we need the following auxiliary result. 

\begin{lemma}[\cite{BDHMT}] \label{lem:1}
Let $\Delta$ be the 
morphism $\KY\longrightarrow\serie{A}{Y^*\otimes Y^*}$ defined on the letters by\footnote{If $\varphi$ is dualizable, this expression can be written
$$
\Delta(y_s)=y_s\otimes 1+1\otimes y_s+\delta(y_s)\ .
$$ 
}
$$
\Delta(y_s)=y_s\otimes 1+1\otimes y_s+\Sum_{n,m\in I} \gamma_{y_n,y_m}^{y_s}y_n\otimes y_m\ .
$$
Then 
\begin{enumerate}
\item for all $u,v,w\in Y^*$, $\scal{u\stuffle_\varphi v}{w}=\scal{u\otimes v}{\Delta(w)}^{\otimes2}$.
\item for all $w\in Y^+$, $\Delta(w)=w\otimes 1+1\otimes w+\Sum_{u,v\in Y^+}\scal{\Delta(w)}{u\otimes v}u\otimes v$.
\end{enumerate}
\end{lemma}

\begin{proof}[Proof of Theorem~\ref{delta_phi} (sketch)] 
The theorem follows by application of items (1) and (2) 
in Lem\-ma~\ref{lem:1}.
\end{proof}

If $\varphi$ is associative (which is fulfilled in all cases of Table~\ref{tab1}), we extend $\varphi$ to $Y^+$ 
by the universal property of the free semigroup $Y^+$,
\begin{equation}
\begin{cases}
\varphi(x)=x,
&\mbox{for }x\in Y,\\
\varphi(xw)=\varphi(x,\varphi(w)),&\mbox{for }x\in Y\mbox{ and }w\in Y^+ ,
\end{cases}
\end{equation}
and we extend the definition of the structure constants accordingly: 
for $x_1\ldots x_l\in Y^+$,
\begin{equation}
\gamma_{x_1\ldots x_l}^{y}=\scal{y}{\varphi(x_1\ldots x_l)}
=\sum_{t_1,\ldots,t_{l-2}\in Y}\gamma^{y}_{x_1,t_1}\gamma^{t_1}_{x_2,t_2}\ldots\gamma^{t_{l-2}}_{x_{l-1},x_l}.
\end{equation}

Note that the fact that $\varphi$ is dualizable can be rephrased as:
\begin{equation}\label{2dual}
\text{for all }y\in Y:\ \{w\in Y^2| \gamma_w^{y}\not=0\}\ \mbox{is finite}.
\end{equation}
In this case, it can be checked immediately that, for an arbitrarily fixed $N\geq 1$,  
\begin{equation}
\text{for all }y\in Y:\ \{w\in Y^N| \gamma_w^{y}\not=0\}\ \mbox{is finite},
\end{equation}
but, by no means, we have in general that
\begin{equation}\label{moderate}
\text{for all }y\in Y:\ \{w\in Y^+| \gamma_w^{y}\not=0\}\ \mbox{is finite}.
\end{equation}
\begin{remark}
i) Condition (\ref{moderate}) is strictly stronger than (\ref{2dual})
as the example of any group law on $Y$, with $|Y|\geq 2$ and finite,
shows.

\medskip
ii) Non-dualizable laws occur with the alphabet
$Y_{\mathbb{Z}}=\{y_j\}_{j\in \mathbb{Z}}$ and the stuffle on it\break
($\varphi(y_i,y_j)=y_{i+j}$). This alphabet naturally appears in the
theory of polylogarithms at negative integers 
in \cite{Ngo} where another non-dualizable law (called $\top$) arises. See also Example~\ref{ex_struct} below.      
\end{remark}
\begin{definition}
An associative law $\varphi$ on $\K Y$ will be 
said to be {\em moderate} if and only if it fulfils condition \eqref{moderate}. 
\end{definition}
     
Let us now state the structure theorem from \cite{BDHMT}.

\begin{theorem}[\cite{BDHMT}]\label{structure}
Let us suppose that $\varphi$ is dualizable and associative. We still denote its dual co-multiplication by
\begin{equation*}
\Delta_{\ministuffle_{\varphi}} : \KY\longrightarrow \KY\otimes \KY.
\end{equation*}
Then
$\mathcal{B}_\varphi=(\KY, {\tt conc}, 1_{Y^*}, \Delta_{\ministuffle_{\varphi}},\varepsilon) $ is a bialgebra.
If, moreover, $\varphi$ is commutative, the following conditions are equivalent:
\begin{enumerate}
\item $\mathcal{B}_\varphi$ is an enveloping bialgebra.
\item $\calB_{\varphi}$ is isomorphic to $(\KY,{\tt conc},1_{Y^*},\Delta_{\shuffle},\epsilon)$ as a bialgebra.
\item For all $y\in Y$, the 
following series is a polynomial
\begin{eqnarray*}
(P)&&y+\Sum_{l\ge2}\Frac{(-1)^{l-1}}{l}\Sum_{x_1,\ldots,x_l\in Y}\scal{y}{\varphi(x_1\ldots x_l)}\;{x_1\ldots x_l}\ .
\end{eqnarray*} 
\item $\varphi$ is moderate.
\end{enumerate}
\end{theorem}
\begin{proof}[Proof (sketch)]

\noindent
4$\ \Longrightarrow\ $3) Obvious.

\medskip
3$\ \Longrightarrow\ $2) One first constructs an endomorphism of
$(\KY, {\tt conc}, 1_{Y^*})$ sending each letter $y\in Y$ to the
polynomial form $(P)$ and 
then proves that it is an automorphism of AAU\footnote{Abbreviation for associative algebra with unit.} which sends $(\KY, {\tt conc}, 1_{Y^*}, \Delta_{\ministuffle_{\varphi}},\varepsilon)$ to $(\KY,{\tt conc},1_{Y^*},\Delta_{\shuffle},\epsilon)$.

\medskip
2$\ \Longrightarrow\ $1) 
Because
 $(\KY,{\tt conc},1_{Y^*},\Delta_{\shuffle},\epsilon)$ is an enveloping bialgebra.

\medskip
1$\ \Longrightarrow\ $4) Observe that, for each letter $y\in Y$, we have 
\begin{equation*}
\scal{\Delta_{\ministuffle_{\varphi}}^{(n-1)}(y)}{x_1\otimes x_2\otimes\cdots \otimes x_n}=\gamma_{x_1\ldots x_l}^{y}\ .
\qedhere
\end{equation*}
\end{proof}
\begin{example}\label{ex_struct}
\begin{enumerate}
\item The muffle product (see Table \ref{tab1}), which determines the product of Hurwitz polyzetas with rational centers and correspond to $\varphi(x_i,x_j)=x_{i.j}$ for $i,j\in\Q_+^*$, is not dualizable ($\gamma_{n,{1}/{n}}^1=1$ for all $n\geq 1$).
\item\label{qbelong} The $q$-infiltration bialgebra (see again Table \ref{tab1}) has its origin in computer science \cite{infiltr1, infiltr2}
and appears as a generic solution in \cite{direct_dual}. It provides a bialgebra 
\begin{equation*}
\calH_{q-\mathrm{inflitr}}=(\ncp{\K}{Y},{\tt conc},1_{X^*},\Delta_{\uparrow_q},\epsilon)
\end{equation*}
($q\in\K$) based on a $\varphi$
which is an associative, commutative and dualizable law, but,
if $Y\not=\emptyset$ this law is moderate only if and only if $q$ is nilpotent in the $\Q$-algebra $\K$. Indeed, for all $x\in Y$, $(1+qx)$ is group-like and it has an inverse in $\ncp{\K}{X}$ if and only if $q$ is nilpotent. In this case the antipode is the involutive antiautomorphism defined on the letters by 
\begin{equation*}
S(x)=\frac{-x}{1+qx}.
\end{equation*}    
\end{enumerate}
\end{example}
\subsection{Structural properties}


Here, we only assume that $\varphi$ is associative.

The bialgebra 
\begin{equation}
\calH_{{\ministuffle_{\varphi}}}^{\vee}=(\KY,{\ministuffle_{\varphi}},1_{Y^*},\Delta_{\tt conc},\epsilon)
\end{equation}    
is a Hopf algebra because it is co-nilpotent\footnote{The law $\Delta_{\tt conc}$, dual to the concatenation is, of course, defined by 
$$
\Delta_{\tt conc}(w)=\sum_{uv=w}\, u\otimes v.
$$
The corresponding $n$-fold $\Delta_{\tt conc}^+$ ($\Delta^+=\Delta$ minus the primitive part) reads 
$$
\Delta_{\tt conc}^{+(n-1)}(w)=\underset{u_i\in Y^+}{\sum_{u_1u_2\cdots u_n=w}}\, u_1\otimes u_2\otimes\cdots \otimes u_n,
$$
from which it is clear that $\Delta_{\tt conc}^{+(n-1)}(w)=0$ for $n> |w|$.}. Its antipode can be computed by $a(1_{Y^*})=1$ and, for $w\in Y^+$,
\begin{equation}\label{antipode1}
a_{\ministuffle_{\varphi}}(w)=\sum_{k\geq 1}(-1)^{-k}
\underset{u_1\ldots u_k=w}{\sum_{u_1,\ldots,u_k\in Y^+}}u_1\stuffle_\varphi\cdots\stuffle_\varphi u_k.
\end{equation}
Due to the finite number of decompositions of any word $u_1\ldots
u_k=w\in Y^+$ into factors $u_1,\ldots,u_k\in Y^+$, we can, at this
very early stage, define an endomorphism $\Phi(S)$ of $\ncp{\K}{Y}$ 
as follows: 
\begin{equation}
\Phi(S)[w]=\sum_{k\geq 1}a_k\underset{u_1\ldots u_k=w}{\sum_{u_1,\ldots,u_k\in Y^+}}u_1\stuffle_\varphi\cdots\stuffle_\varphi u_k,
\end{equation}
associated to any univariate formal power series $S=a_1X+a_2X^2+a_3X^3+\cdots$. The case of 
\begin{equation}
\log(1+X)=\sum_{k\geq 1}\frac{(-1)^{k-1}}{k}X^k
\end{equation} 
will be of particular importance. It reads here in the style of formula (\ref{antipode1}),
\begin{equation}\label{check-pi-1}
\check\pi_1(w)=\sum_{k\geq 1}\frac{(-1)^{k-1}}{k}
\underset{u_1\ldots u_k=w}{\sum_{u_1,\ldots,u_k\in Y^+}}u_1\stuffle_\varphi\cdots\stuffle_\varphi u_k.
\end{equation} 
This $\check\pi_1\in \mathrm{End}(\ncp{\K}{Y})$ has an adjoint $\pi_1\in \mathrm{End}(\ncs{\K}{Y})$ which reads
\begin{align}\label{pi-1-first}
\pi_1(S)&=\sum_{w\in Y^*}\scal{S}{\check\pi_1(w)}w\\
&=\sum_{k\geq 1}\frac{(-1)^{k-1}}{k}\sum_{u_1,\ldots,u_k\in Y^+}
\scal{S}{u_1\stuffle_\varphi\cdots\stuffle_\varphi u_k}\, u_1\ldots u_k.
\end{align}
It is an easy exercise to check that the family in the sums of (\ref{pi-1-first}) is summable\footnote{A family of (simple, double, etc.) series is summable if it is locally finite (see \cite{BDHMT} for a complete development).}. 
It is easy to check that the dominant term of all terms in a
$\stuffle_\varphi$ product is the corresponding $\shuffle$ product. 
This 
explains why we still have the theorem of Radford.
\begin{theorem}[\sc Radford's theorem]
When $\varphi$ is commutative, the associative and commutative algebra
with unit $(\ncp{\K}{Y},\ministuffle_{\varphi},1_{Y^*})$ is a
polynomial algebra. More precisely, the 
morphism $\beta : \K[\Lyn Y]\rightarrow
(\ncp{\K}{Y},\ministuffle_{\varphi},1_{Y^*})$ defined by $\beta(l)=l$ for $l\in \Lyn Y$ is an isomorphism. In other words, the family 
\begin{equation*}
\Big(l_1^{\ministuffle_{\varphi} i_1}\ministuffle_{\varphi}\cdots
\ministuffle_{\varphi} l_k^{\ministuffle_{\varphi} i_k}\Big)
_{\smallmatrix{k\geq 0,\ \{l_1,l_2,\dots ,l_k\}\subset \Lyn Y}\\
{(i_1,i_2,\dots ,i_k)\in (\N_+)^k}\endsmallmatrix}   
\end{equation*}   
is a linear basis of $\ncp{\K}{Y}$.
\end{theorem}

\begin{proof}
One checks that 
$$
l_1^{\ministuffle_{\varphi} i_1}\ministuffle_{\varphi}\cdots
\ministuffle_{\varphi} l_k^{\ministuffle_{\varphi} i_k}=l_1^{\minishuffle i_1}\minishuffle\ldots
\minishuffle l_k^{\minishuffle i_k}+\sum_{|v|<\sum_{1\leq j\leq k}i_j|l_j|} c_v\, v.
$$
The result follows.
\end{proof}
The theorem of Radford is important in the classical cases because it
is the left-hand side of Sch\"utzenberger's factorization in which we have the move
$$\mbox{PBW}\rightarrow\mbox{Radford};$$
see \cite{DDM} for a discussion of the converse.
\begin{lemma}[\sc $\varphi$-extended Friedrichs' criterion]\label{Friedrichs}
We denote\footnote{As in the classical case, $\Delta_{\ministuffle_{\varphi}}$ is a ${\tt conc}$-morphism as can be seen by transposition of the fact that $\Delta_{\tt conc}$ is a $\ministuffle_{\varphi}$-morphism.} by
\begin{equation*}
\Delta_{\ministuffle_{\varphi}} : \ncs{\K}{Y}\rightarrow \ncs{\K}{Y^*\otimes Y^*}
\end{equation*}
the dual of $\ministuffle_{\varphi}$ applied to series, i.e., defined by 
\begin{equation*}
\Delta_{\ministuffle_{\varphi}}(S)=\sum_{u,v\in Y^*}\scal{S}{u\ministuffle_{\varphi}\, v}\, u\otimes v.
\end{equation*}
Let now $S\in\ncs{\K}{Y}$. Then we have:
\begin{enumerate}
\item If $\scal{S}{1_{Y^*}}=0$ then $S$ is primitive
(i.e.,
  $\Delta_{\ministuffle_{\varphi}}(S)=S\otimes1_{Y^*}+1_{Y^*}\otimes
  S$)\footnote{Tensor products of linear forms.} if and only if we
  have
$\scal{S}{u\ministuffle_{\varphi} v}=0$ for any $u$ and $v\in Y^+$.
\item If $\scal{S}{1_{Y^*}}=1$, then $S$ is group-like (i.e., $\Delta_{\ministuffle_{\varphi}}(S)=S\otimes S$)\footnote{idem}
if and only if we have $\scal{S}{u\ministuffle_{\varphi}
  v}=\scal{S}{u}\scal{S}{v}$ for any $u$ and $v\in Y^*$.
\end{enumerate}
\end{lemma}
\begin{proof}
The expected equivalences are due to the following facts:
\begin{align*}
\Delta_{\ministuffle_{\varphi}}(S)&=S\otimes1_{Y^*}+1_{Y^*}\otimes 
S-\scal S{1_{Y^*}}1_{Y^*}\otimes1_{Y^*}+\sum_{u,v\in Y^+}
\scal S{u\ministuffle_{\varphi} v}u\otimes v,\\
\Delta_{\ministuffle_{\varphi}}(S)&=\sum_{u,v\in Y^*}\scal S{u\ministuffle_{\varphi} v}u\otimes v
\quad\mbox{and}\quad S\otimes S=\sum_{u,v\in Y^*}\scal{S}{u}\scal{S}{v}u\otimes v.
\qedhere
\end{align*}
\end{proof}
\begin{lemma}\label{lemme2} 
Let $S\in\ncs{\K}{Y}$ be such that $\scal{S}{1_{Y^*}}=1$. Then $S$ is group-like if and only if\,\footnote{For any $h\in\serie{\K}{Y}$, if $\scal{h}{1_{Y^*}}=0$, we define
\begin{eqnarray*}
\log(1_{Y^*}+h)=\sum_{n\ge1}\frac{(-1)^{n-1}}{n}h^n&\mbox{and}&\exp(h)=\sum_{n\ge0}\frac{h^n}{n!},
\end{eqnarray*}
and we have the usual formulas $\log(\exp(h))=h$ and $\exp(\log(1_{Y^*}+h))=1_{Y^*}+h$.} 
$\log(S)$ is primitive.
\end{lemma}
\begin{proof}
Since $\Delta_{\ministuffle_{\varphi}}$ and the maps  $T\mapsto T\otimes1_{Y^*},T\mapsto1_{Y^*}\otimes T$
are continuous homomorphisms, then, if $\log S$ is primitive, we have (see Lemma~\ref{Friedrichs}(1))
\begin{equation*}
\Delta_{\ministuffle_{\varphi}}(\log S)=\log S\otimes1_{Y^*}+1_{Y^*}\otimes\log S,
\end{equation*}
and, since $\log S\otimes1_{Y^*}$ and $1_{Y^*}\otimes\log S$ commute,
we get successively
\begin{align*}
\Delta_{\ministuffle_{\varphi}}(S)&=\Delta_{\ministuffle_{\varphi}}(\exp(\log S))\\
&=\exp(\Delta_{\ministuffle_{\varphi}}(\log S))\\
&=\exp(\log S\otimes1_{Y^*})\exp(1_{Y^*}\otimes\log S)\\
&=(\exp(\log S)\otimes1_{Y^*})(1_{Y^*}\otimes\exp(\log S))\\
&=S\otimes S.
\end{align*}
This means, together with $\scal{S}{1_{Y^*}}$, that $S$ is group-like. The converse can be obtained in the same way.
\end{proof}
\begin{remark}
i) In fact, Lemma~\ref{lemme2} establishes a nice log-exp
correspondence for the Lie group of group-like series.

\medskip
ii) Through the canonical pairing $\scal{-}{-}:\ncs{\K}{Y}\otimes
\ncp{\K}{Y}\to \K$, we have $\ncs{\K}{Y}\simeq (\ncp{\K}{Y})^*$. Group-like (respectively primitive) series are in bijection with characters (respectively infinitesimal characters) of the algebra $(\ncp{\K}{Y},\ministuffle_{\varphi},1_{Y^*})$.
\end{remark} 
\begin{lemma}\label{lemme3} 
\begin{enumerate}
\item Group-like series form a group (for concatenation).  
\item The space $\mathrm{Prim}(\ncs{\K}{Y})$ is a Lie algebra (for the bracket derived from concatenation). 
\end{enumerate}
\end{lemma}
\begin{proof}
As in the classical case. 
\end{proof}
We extend the transposition process in the same way as in
Lemma~\ref{Friedrichs} and note, for $n\geq 1$, that
\begin{equation}
\Delta_{\ministuffle_{\varphi}}^{(n-1)} : \ncs{\K}{Y}\rightarrow \ncs{\K}{(Y^*)^{\otimes\,n}},
\end{equation}
the dual of $\ministuffle_{\varphi}^{(n-1)}$ applied to series, i.e., defined by 
\begin{equation}
\Delta_{\ministuffle_{\varphi}}^{(n-1)}(S)=\sum_{u_1,u_2,\dots u_n\in Y^*}\scal{S}{u_1\ministuffle_{\varphi}\cdots  \ministuffle_{\varphi}u_n}\, u_1\otimes\cdots \otimes u_n \ .
\end{equation}
We will use the following lemma several times, which gives the combinatorics of products of primitive series (and the polynomials).
\begin{lemma}[\sc Higher order co-multiplications of products]\label{language_lemma}
Let us consider the  language $\mathcal{M}$ over the alphabet $\calA=\{a_1,a_2,\ldots,,a_m\}$,
\begin{equation*}
\mathcal{M}=\{w\in\calA^*\mid w=a_{j_1}\dots a_{j_{|w|}},j_1<\dots<j_{|w|},1\leq |w|\leq m\},
\end{equation*}
and the 
morphism
\begin{eqnarray*}
\mu:\K\langle\calA\rangle&\longrightarrow&\ncs{\K}{Y},\\
a_i&\longmapsto&S_i,
\end{eqnarray*}
where $S_1,\ldots,S_m$ are primitive series in $\ncs{\K}{Y}$. Then
\begin{equation*}
\Delta_{\ministuffle_{\varphi}}^{(n-1)}(S_1\ldots S_m)=
\underset{a_1\cdots a_m\in\supp(w_1\shuffle\ldots\shuffle w_n)}
{\underset{|w_1|+\dots+|w_n|=m}{\sum_{w_1,\ldots,w_n\in\mathcal{M}}}}
\mu(w_1)\otimes\cdots\otimes\mu(w_n).
\end{equation*}
\end{lemma}
\begin{proof}[Proof (sketch)]
Let $\mathcal{S}=(S_1,\dots,S_m)$ be this family of primitive series
and, for $I=\{i_1,\break\dots, i_k\}\subset [1\cdots m]$ in increasing
order, let us write $\mathcal{S}[I]$ for 
the product $S_{i_1}\cdots S_{i_k}$. Then we have
\begin{equation*}
\Delta_{\ministuffle_{\varphi}}^{(n-1)}(S_1\ldots S_m)=
\sum_{I_1+\cdots +I_n=[1\cdots m]}\mathcal{S}[I_1]\otimes\cdots \otimes \mathcal{S}[I_n].
\end{equation*}
Setting $w_i=(a_1a_2\ldots a_m)[I]$, one gets the expected result.
\end{proof}
\begin{lemma}[\sc Pairing of products]
Let $S_1,\ldots,S_m$ be primitive series in $\ncs{\K}{Y}$, and let $P_1,\ldots,P_n$
be {\em proper}\footnote{i.e., polynomials without constant term; see
  \cite{berstel}.} polynomials in $\KY$. Then 
one has in general
\begin{equation*}
\scal{P_1\ministuffle_{\varphi}\cdots\ministuffle_{\varphi}P_n}{S_1\ldots S_m}
=
\underset{a_1\cdots a_m\in\supp(w_1\shuffle\ldots\shuffle w_n)}
{\underset{|w_1|+\dots+|w_n|=m}{\sum_{w_1,\ldots,w_n\in\mathcal{M}
}}}\prod_{i=1}^n\scal{P_i}{\mu(w_i)}.
\end{equation*}
In particular, we have the following exhaustive list of 
circumstances:
\begin{enumerate}
\item  If $n>m$, then  $\scal{P_1\ministuffle_{\varphi}\cdots\ministuffle_{\varphi}P_n}{S_1\ldots S_m}=0.$
\item  If $n=m$, then
\begin{equation*}
\scal{P_1\ministuffle_{\varphi}\cdots\ministuffle_{\varphi}P_n}{S_1\ldots S_n}
=\sum_{\sigma\in\mathfrak{S}_n}\prod_{i=1}^n\scal{P_i}{S_{\sigma(i)}}.
\end{equation*}
\item
If $n<m$, 
then one has the general form in which every product in the sum
contains at least a factor $\scal{P_i}{\mu(w_i)}$ with $|w_i|\geq 2$.
\end{enumerate}
\end{lemma}

\begin{proof}
It is a consequence of Lemma~\ref{language_lemma} 
through the equality
\begin{equation*}
\scal{P_1\ministuffle_{\varphi}\cdots\ministuffle_{\varphi}P_n}{S_1\ldots S_n}=
\scal{P_1\otimes\cdots\otimes P_n}{\Delta_{\ministuffle_{\varphi}}^{(n-1)}(S_1\ldots S_n)}\ .
\qedhere
\end{equation*}
\end{proof}
In the sequel, we assume that $\varphi$ is associative and dualizable.

\bigskip
Now, we have the following two structures:
\begin{align}
\calH_{{\ministuffle_{\varphi}}}&=(\KY,{\tt conc},1_{Y^*},\Delta_{\ministuffle_{\varphi}},\epsilon),\\
\calH_{{\ministuffle_{\varphi}}}^{\vee}&=(\KY,{\ministuffle_{\varphi}},1_{Y^*},\Delta_{\tt conc},\epsilon),
\end{align}    
which are mutually dual\footnote{This duality is separating; see \cite{TVS}.} bialgebras. 
The bialgebra $\calH_{{\ministuffle_{\varphi}}}$ need not be a Hopf algebra, {\it even if $\Delta_{\ministuffle_{\varphi}}$ is cocommutative} (see Example~\ref{ex_struct}.2). 

Now, let us consider
\begin{align}
\calI&:=\mathrm{span}_{\K}\{u\ministuffle_{\varphi}v\}_{u,v\in Y^+},\\
\K_+\pol{Y}&:=\{P\in\KY\bv\scal{P}{1_{Y^*}}=0\},\\
\calP&:=\mathrm{Prim}(\calH_{{\ministuffle_{\varphi}}})=\{P\in\KY\bv\Delta^+_{\ministuffle_{\varphi}}(P)=0\},
\end{align}
where
\begin{equation}
\Delta^+_{\ministuffle_{\varphi}}(P)=\Delta_{\ministuffle_{\varphi}}(P)-(P\otimes1_{Y^*}+1_{Y^*}\otimes P)+
\scal{P}{1_{Y^*}}1_{Y^*}\otimes 1_{Y^*}\ .
\end{equation}
\begin{remark}
At this stage ($\varphi$ not necessarily moderate), it 
can happen that 
$\mathrm{Prim}(\calH_{{\ministuffle_{\varphi}}})=\{0\}$.
This is for example the case 
with the $q$-infiltration bialgebra on one letter at $q=1$,
\begin{equation*}
\calH_{{\ministuffle_{\varphi}}}=(\K[x],{\tt conc},1_{x^*},\Delta_{\uparrow_1},\epsilon),
\end{equation*} 
and, more generally, when 
$q$ is not nilpotent\footnote{Recall that $q$ is an element of the ring $\K$ (see example \ref{ex_struct}.\ref{qbelong}).}.
\end{remark}
We can also endow $\mathrm{End}(\KY)$, the $\K$-module of endomorphisms of $\KY$, with the {\it convolution} product defined by
\begin{eqnarray}
\text{for all }f,g\in\mathrm{End}(\KY),&&f\star g=\conc\circ(f\otimes g)\circ\Delta_{\ministuffle_{\varphi}},\\
\mbox{i.e.},
\text{for all }P\in\KY,&&(f\star g)(P)=\sum_{u,v\in Y^*}\scal{P}{u\ministuffle_{\varphi}v}f(u)g(v).
\end{eqnarray}
Then $\mathrm{End}(\KY)$ becomes a $\K$-associative algebra with unity (AAU), its unit being $e={1}_{\KY}\circ\epsilon$.

\smallskip
It is convenient to represent every $f\in \mathrm{End}(\KY)$ by its graph, a double series which reads 
\begin{equation}
\Gamma(f)=\sum_{w\in Y^*}w\otimes f(w).
\end{equation}
This representation is faithful and, by direct computation, one gets 
\begin{equation}
\Gamma(f)\Gamma(g)=\Gamma(f\star g),
\end{equation}
where the multiplication of double series is performed by the stuffle on the left and the concatenation on the right. 
\begin{definition}
Here $t$ is a real parameter. 
Let us define define
\begin{eqnarray*}
\calD_Y:=\Gamma(\mathrm{Id}_{\ncp{\K}{Y}})=\sum_{w\in Y^*}w\otimes w,&\Haus_Y:=\log\calD_Y,&\sigma_Y(t):=\exp(t\Haus_Y).
\end{eqnarray*}
\end{definition}

\bigskip
From now on, we assume that $\varphi$ is associative, commutative and dualizable.
\begin{lemma}[\sc $\pi_1$ is a projector on the primitive series]
\label{primitive_proj}
The endomorphism $\pi_1$ is a projector, the image of which is exactly 
the space of primitive series, $\mathrm{Prim}(\ncs{\K}{Y})$.  
\end{lemma}

\begin{proof}[Proof (sketch)]
The proof follows
the lines of \cite{reutenauer} with the difference that $\pi_1(w)$  might not be a polynomial and the operator defined in Lemma~\ref{Friedrichs} is not a genuine co-product. The diagonal series $\calD_Y$ (when considered as a series in $\ncs{(\ncs{\K}{Y})}{Y}$, the coefficient ring, $\ncs{\K}{Y}$, being endowed with the $\ministuffle_{\varphi}$ product) is group-like in the sense of Lemma~\ref{Friedrichs}. Then, using 
$$\log(\calD_Y)=\sum_{w\in Y^*}w\otimes \pi_1(w)$$
(which can be established by summability of the family $(w\otimes
\pi_1(w))_{w\in Y^*}$; but 
remember that the $\pi_1(w)$ are, in general,
series\footnote{In greater  
detail, this equality amounts to checking the summability of the family 
$$
\Big(\frac{(-1)^{k-1}}{k}w\otimes\scal{w}{u_1\ministuffle_{\varphi}\cdots \ministuffle_{\varphi}u_k}u_1\cdots u_k
\Big)_{\smallmatrix w\in Y^*,\ k\geq 1\\ u_1,\dots,u_k\in Y^+\endsmallmatrix}
$$
 (which is immediate) and rearranging the sums.}), one gets that 
$\pi_1(w)$ is a primitive series for all $w$. Now, from
$$\pi_1(S)=\sum_{w\in Y^*}\scal{S}{w}\pi_1(w),$$
one has $\pi_1(S)\in \mathrm{Prim}(\ncs{\K}{Y})$. Conversely, from
Friedrichs' criterion, 
one gets $\pi_1(S)=S$ if $S\in \mathrm{Prim}(\ncs{\K}{Y})$. 
\end{proof}

\bigskip
In the remainder of the paper, we suppose that $\varphi$ is moderate
(and still dualizable, associative and commutative).
\begin{definition}[\sc Projectors, \cite{reutenauer}]\label{Eulerian} 
Let $\Iplus:\KY\longrightarrow\KY$ be the linear mapping defined by
\begin{eqnarray*}
\Iplus(1_{Y^*})=0,&\mbox{and}&\text{for all }w\in Y^+,\ \Iplus(w)=w.
\end{eqnarray*}
One defines\footnote{The series below are summable because the family $(\Iplus^{\star n})_{n\geq 0}$ is locally nilpotent (see \cite{BDHMT} for complete proofs). Note that this definition gives the same result as the computation of the adjoint of $\check\pi_1$ given in \eqref{pi-1-first}.}
\begin{eqnarray*}
\pi_1:=\log(e+\Iplus)=\Sum_{n\ge1}\Frac{(-1)^{n-1}}n\Iplus^{\star n},&\mbox{where}&
\Iplus^{\star n}:=\conc_{n-1}\circ\Iplus^{\otimes n}\circ\Delta^{(n-1)}_{\ministuffle_{\varphi}}.
\end{eqnarray*}
\end{definition}

It follows immediately that
\begin{equation}
\exp(\pi_1)=e+\Sum_{n\ge1}\Frac{1}{n!}\pi_1^{\star n}=\sum_{n\ge0}\pi_n,
\end{equation}
where $e={1}_{\KY}\circ\epsilon$ is the orthogonal complement of $\Iplus$ and neutral for the convolution product.
The $\pi_n$ so obtained is called the $n$-th Eulerian projector.
\begin{lemma}
The endomorphism $\check\pi_1$ defined in \eqref{check-pi-1} is the
adjoint of $\pi_1$. 
One has
\begin{equation*}
\check\pi_1=\Sum_{n\ge1}\Frac{(-1)^{n-1}}n\ministuffle_{\varphi}^{(n-1)}\circ\Iplus^{\otimes n}\circ\Delta^{(n-1)}_{\tt conc}.
\end{equation*}
\end{lemma}
\begin{proof}
Immediate.
\end{proof}
\begin{proposition}\label{L2bis} 
{\sc (Graph of $\pi_1$, values and its exponential
    as resolution of\break unity)}{\bf.}

\begin{enumerate}
\item For all $Y$ and $\varphi$ (moderate, associative,
  commutative and dualizable), 
one has
\begin{equation*}
\log\calD_Y=\sum_{w\in Y^+}w\otimes\pi_1(w)=\sum_{w\in Y^+}\check\pi_1(w)\otimes w.
\end{equation*}
\item Let $P\in\KY$ be a primitive polynomial, for $\Delta_{\ministuffle_{\varphi}}$. Then
\begin{eqnarray*}
\pi_1(P)=P,&\text{for all }k,n\in\N_+,&\pi_n(P^k)=\delta_{k,n}P^k.
\end{eqnarray*}
\item 
One has
\begin{equation}\label{eq3}
\mathrm{Id}_{\ncp{\K}{Y}}=e+I_+=\sum_{n\geq 0}\pi_n.
\end{equation}
Equation \eqref{eq3} is a resolution of identity with mutually orthogonal summands.
\item\label{4} We have
\begin{equation*}
\K_+\pol{Y}=\calP{\overset{\perp}{\oplus}}\calI=\calP\oplus\Big(\bigoplus_{n\geq 2}\pi_n(\ncp{\K}{Y})\Big).
\end{equation*}
\end{enumerate}
\end{proposition}
\begin{proof}
The only statement which cannot be proved through an isomorphism with
the shuffle algebra is the first equality of 
the point~(4).
The fact that $\calP\cap\calI=\{0\}$ comes from Friedrichs' criterion, and $\calP+\calI=\K_+\pol{Y}$ is a consequence of the fact
(seen again through any isomorphism with the shuffle algebra) that 
\begin{equation*}
(\calH_{{\ministuffle_{\varphi}}})_+=\mathrm{span}_{\K}(\bigcup_{n\geq 1}
{(P_1\ministuffle_{\varphi}\cdots\ministuffle_{\varphi}P_n)}_{P_i\in\mathrm{Prim}(\calH_{{\ministuffle_{\varphi}}})}).
\qedhere
\end{equation*} 
\end{proof}
\begin{remark}\label{ree_thm1}
\begin{enumerate}
\item The first equality of Proposition~\ref{L2bis}.(4), i.e.,
\begin{equation*}
\K_+\pol{Y}=\calP{\overset{\perp}{\oplus}}\calI,
\end{equation*}
is known as the theorem of Ree \cite{Ree}.
\item The projector on $\calP$ parallel to $\calI$ is not in general in the descent algebra (see \cite{ncsf3}).
This proves that, although they are isomorphic, the spaces $\calI$ and $\bigoplus_{n\geq 2}\pi_n(\ncp{\K}{Y})$ are, in general, not identical.  
\end{enumerate}
\end{remark}
Proposition~\ref{L2bis}.(1) leads to 
the following corollary.
\begin{corollary}
We have $\pi_1(1_{Y^*})=\check\pi_1(1_{Y^*})=0$ and, for all $w\in Y^+$,
\begin{align*}
\pi_1(w)&=w+\sum_{k\ge2}\frac{(-1)^{k-1}}k\sum_{u_1,\ldots,u_k\in Y^+}\scal{w}{u_1\ministuffle_{\varphi}\cdots\ministuffle_{\varphi}u_k}u_1\ldots u_k,\\
{\check\pi_1}(w)&=w+\sum_{k\ge2}\frac{(-1)^{k-1}}k\sum_{u_1,\ldots,u_k\in Y^+}\scal{w}{u_1\ldots u_k}u_1\ministuffle_{\varphi}\cdots\ministuffle_{\varphi} u_k.
\end{align*}
In particular, $\pi_1(1_{Y^*})=\check\pi_1(1_{Y^*})=0$, for any $y\in
Y$, ${\check\pi_1}(y)=y$, and
\begin{equation*}
\pi_1(y)=y+\sum_{l\ge2}\frac{(-1)^{l-1}}l\sum_{x_1\ldots x_l\in Y^*}\gamma^{y}_{x_1,\ldots,x_l}\;{x_1\ldots x_l}.
\end{equation*}
\end{corollary}
\begin{remark}
We already knew that, as soon as $\varphi$ is associative,
${\check\pi_1}(w)$ is a polynomial. Here, because $\varphi$ is
moderate, dualizable, and associative, $\pi_1(w)$ is also a
polynomial, and because $\varphi$ is commutative, 
it is primitive.  
\end{remark}

\begin{proposition}\label{pi_1}
We have:
\begin{enumerate}
\item The expression of $\sigma_Y(t)$ is given by
\begin{equation*}
\sigma_Y(t)=\Sum_{n\ge0}{t^n}\Sum_{w\in Y^*}w\otimes\pi_n(w)
=\Sum_{n\ge0}{t^n}\Sum_{w\in Y^*}{\check\pi_n}(w)\otimes w,
\end{equation*}
where $\check\pi_n$ is the adjoint of $\pi_n$. These are given by $\pi_n(1_{Y^*})=\check\pi_n(1_{Y^*})=\delta_{0,n}$ and, for all $w\in Y^+$,
\begin{align*}
\pi_n(w)&=\Frac1{n!}\Sum_{u_1,\ldots,u_n\in Y^+}\scal{w}{\check\pi_1(u_1)\ministuffle_{\varphi}\cdots\ministuffle_{\varphi}\check\pi_1(u_n)}\pi_1(u_1)\ldots\pi_1(u_n),\\
\check\pi_n(w)&=\Frac1{n!}\Sum_{u_1,\ldots,u_n\in Y^+}\scal{w}{\pi_1(u_1)\ldots\pi_1(u_n)}\check\pi_1(u_1)\ministuffle_{\varphi}\cdots\ministuffle_{\varphi} \check\pi_1(u_n)).
\end{align*}

\item  For any $w\in Y^*$, we  have
\begin{align*}
w&=\sum_{k\ge0}\frac1{k!}\sum_{u_1,\ldots,u_k\in Y^+}\scal{w}{u_1\ministuffle_{\varphi}\cdots\ministuffle_{\varphi} u_k}\pi_1(u_1)\ldots\pi_1(u_k)\\
&=\sum_{k\ge0}\frac1{k!}\sum_{u_1,\ldots,u_k\in Y^+}\scal{w}{u_1\ldots u_k}{\check\pi_1}(u_1)\ministuffle_{\varphi}\cdots\ministuffle_{\varphi}{\check\pi_1}(u_k).
\end{align*}
\end{enumerate}
In particular, for any $y_s\in Y$, we have $y_s=\check\pi_1(y_s)$ and
\begin{equation*}
y_s=\sum_{k\ge1}\frac1{k!}\sum_{y_{s_1},\ldots,y_{s_k}\in Y}\gamma^{y_s}_{y_{s_1},\ldots,y_{s_k}}\pi_1(y_{s_1})\ldots\pi_1(y_{s_k}).
\end{equation*}
\end{proposition}

\begin{proof} Direct computation.
\end{proof}

Applying the tensor product\footnote{Extended to series.} of isomorphisms of algebras\footnote{In order to clarify the ideas at this point,
the reader can also take the alphabet duplication
isomorphism $$\text{for all }\bar y\in\bar Y,\ \bar y=\alpha(y),$$
and use  $\{w\}_{w\in\bar Y^*}$ as a basis for $\K\pol{\bar Y}$.} $\alpha\otimes\mathrm{Id}_Y$ to the diagonal series $\calD_Y$,
we obtain a group-like element, and then computing the logarithm of this element
(or equivalently, applying $\alpha\otimes\mathrm{Id}_Y$ to $\Haus_Y$) we obtain $\calS$ which is, by Lemma~\ref{lemme2}, primitive:
\begin{equation}\label{calS}
\calS=\sum_{w\in Y^*}\alpha(w)\;\pi_1(w)=\sum_{w\in Y^*}\alpha\circ\check\pi_1(w)\;w.
\end{equation}

\begin{lemma}\label{L2}
For any $w\in Y^+$, 
one has $\pi_1(w)\in\mathrm{Prim}(\ncp{\K}{Y})$.
\end{lemma}

\begin{proof}
Immediate from Lemma~\ref{primitive_proj}.
\end{proof}

A primitive projector, $\pi:\KY\longrightarrow\KY$, is defined in the same way as a Lie projector by the three following conditions:
\begin{eqnarray}
\pi\circ\pi=\pi,&\pi(1_{Y^*})=0,&\pi(\KY)=\mathrm{Prim}(\KY)=\calP.
\end{eqnarray}
For example, $\pi_1$ defined in Definition~\ref{Eulerian} (see also Proposition~\ref{L2bis})
is a primitive projector which will be used to construct bases of $\calP$ and its enveloping algebra (see Theorem~\ref{bases} below). Another example of a primitive projector is the orthogonal projector on $\calP$ attached to the decomposition in Remark~\ref{ree_thm1}.
%

Now, for the remainder of the paper, let  $\calY=\{y_w\}_{w\in Y^+}$ (respectively $\calY_1=\{y_x\}_{x\in Y}$) be a copy of $Y^+$ (respectively $Y$).

Let us then equip $\K\pol{\calY}$ and $\K\pol{\calY_1}$
with $\bullet$ (the concatenation so denoted to be distinguished from the concatenation within $Y^+$) and $\shuffle$ (or equivalently by $\Delta_{\bullet}$ and $\Delta_{\minishuffle}$).

Thus, the Hopf algebras $(\K\pol{\calY},\bullet,1_{\calY^*},\Delta_{\minishuffle},\epsilon_{\calY^*})$
and $(\K\pol{\calY_1},\bullet,1_{\calY^*},\Delta_{\minishuffle},\epsilon_{\calY_1^*})$  
are connected, $\N$-graded, non-commutative and co-commu\-tative
bialgebras, and hence enveloping bialgebras (in fact, they are free
algebras but 
specially indexed to match our purpose).

Now we can state the following      
%
\begin{theorem}[\sc New letters as images]\label{phipi1}
Let $\pi:\KY\longrightarrow\KY$ be a primitive projector. Let  $\psi_{\pi}$ be the $\conc$-morphism defined by
\begin{eqnarray*}
\psi_{\pi}:\K\pol{\calY}&\longrightarrow&\KY,\\
y_w&\longmapsto&\psi_{\pi}(y_w)=\pi(w ).
\end{eqnarray*}
Then $\psi_{\pi}$ is surjective and a Hopf 
morphism.

Moreover, $\ker\psi_{\pi}=\calJ=\calJ_1+\calJ_2$, where $\calJ_1$ and
$\calJ_2$ are the two-sided ideals of $\ncp{\K}{\calY}$ generated by
\begin{equation*}
S_1=\{y_u-y_{\pi(u)}\}_{u\in Y^+}\mbox{ and }S_2=\{y_u\bullet y_v-y_v\bullet y_u-y_{[\pi(u),\pi(v)]}\}_{u,v\in Y^+},
\end{equation*} 
respectively,
where the indexing of the alphabet has been extended by linearity to polynomials, i.e.,
\begin{equation*}
y_{P}:=\sum_{w\in Y^+}\scal{P}{w}\,y_w.
\end{equation*}
\end{theorem}
\begin{proof}
The fact that $\psi_{\pi}$ is surjective is due to
$\pi(\ncp{\K}{Y})=\calP$ and to 
the fact that any enveloping algebra (here
$\calH_{{\ministuffle_{\varphi}}}$) is generated by its primitive
elements. The fact that $\psi_{\pi}$ is a Hopf 
morphism is due to a general property of enveloping algebras:
{\it if a 
morphism of AAU between two enveloping algebras sends the
primitive elements of the first 
to primitive elements of the second,
then it is a Hopf 
morphism}.

Let now $(p_i)_{i\in J}$ be an ordered ($J$ is endowed with a total
ordering $\prec_J$) basis\footnote{With the properties of $\varphi$
  here, the bialgebra $(\KY,{\tt
    conc},1_{Y^*},\Delta_{\ministuffle_{\varphi}},\epsilon)$ is
  isomorphic to $(\KY,{\tt
    conc},1_{Y^*},\Delta_{\minishuffle},\epsilon)$ in which the module
  of primitive elements is free, thus
  $\calP=\mathrm{Prim}(\ncp{\K}{Y})$ is free.} of
$\calP=\mathrm{Prim}(\ncp{\K}{Y})$, and 
let us recall that $\calJ=\calJ_1+\calJ_2$ denotes the two sided ideal generated by the elements $\calJ_i$ (itself generated by $S_i,\ i=1,2$).

First, we observe that the elements of $S_1\cup S_2$ are in the kernel of $\psi_{\pi_1}$, and then $\mathcal{J}\subset \ker\psi_{\pi_1}$.

On the other hand, for $u_1,u_2,\ldots ,u_n\in Y^+$, one has
\begin{equation}\label{projpi}
y_{u_1}\bullet y_{u_2}\bullet\cdots \bullet y_{u_n} \equiv y_{\pi(u_1)}\bullet y_{\pi(u_2)}\bullet\cdots \bullet y_{\pi(u_n)}\mod  \mathcal{J} 
\end{equation}
(in fact they are even equivalent $\hspace{-3mm}\mod  \mathcal{J}_1$), 
which amounts to 
saying that $\ncp{\K}{\calY}=\mathcal{J}+\lessdot \calP \gtrdot$, where $\lessdot \calP \gtrdot$ is the space ``generated by $\calP$'', in fact, generated by 
$$
\bigsqcup_{n\geq 0}\{y_{p_{i_1}}\bullet\cdots  \bullet y_{p_{i_n}}\}_{i_j\in J}.
$$ 
Now, by 
recurrence over the number of inversions, one can show, using $S_2$, that 
\begin{equation}\label{reorder}
 y_{p_{i_1}}\bullet\cdots  \bullet y_{p_{i_n}}\equiv 
 y_{p_{\sigma(i_1)}}\bullet\cdots  \bullet y_{p_{\sigma(i_n)}} \mod  \calJ ,
\end{equation} 
where $\sigma\in \mathfrak{S}_n$ is such that $\sigma(i_1)\succ_J \sigma(i_2)\succ_J   \cdots \succ_J \sigma{(i_n)}$ (large order reordering).

Let $\calC$ be the space generated by the elements 
\begin{equation}
\{y_{p_{j_1}}\bullet\cdots  \bullet y_{p_{j_n}}\}_{\smallmatrix j_1\succ_J
  j_2\succ_J   \cdots \succ_J j_n\\ n\geq 0\endsmallmatrix}\ .
\end{equation}
By (\ref{projpi}) and (\ref{reorder}), we get $\calJ+\calC=\ncp{\K}{\calY}$.

Now, due to the PBW theorem, the family of images 
\begin{equation}
\Big(\Phi_{\pi_1}(y_{p_{j_1}}\bullet y_{p_{j_2}}\bullet\cdots  \bullet
y_{p_{j_n}})\Big)_{\smallmatrix j_1\succ_J j_2\succ_J   \cdots  j_n\\ n\geq
  0\endsmallmatrix}
\end{equation}
is a basis of $\ncp{\K}{Y}$, which proves that $\Phi_{\pi_1}\hspace*{-2mm}\mid_{\calC} : \calC \rightarrow \ncp{\K}{Y}$ is an isomorphism and completely proves the claim.
\end{proof}
We now suppose that the alphabet $Y$ is totally ordered.
\begin{definition}\label{Pi}
\begin{enumerate}
\item Let $\{\Pi_l\}_{l\in \Lyn Y}$ and $\{\Pi_w\}_{w\in Y^*}$ be the families of elements of $\calP$ and $\KY$, respectively, obtained as follows:
\begin{alignat*}2
\Pi_{y_k}&=\pi_1(y_k),&&\mbox{for }k\ge1,\\
\Pi_{l}&=[\Pi_s,\Pi_r],&&\mbox{for }l\in\Lyn X,\mbox{ standard factorization of }l=(s,r),\\
\Pi_{w}&=\Pi_{l_1}^{i_1}\ldots \Pi_{l_k}^{i_k},\quad 
&&\mbox{for }w=l_1^{i_1}\ldots l_k^{i_k},l_1\succ_{lex}\cdots\succ_{lex} l_k,l_1\ldots,l_k\in\Lyn Y.
\end{alignat*}
\item Let $\{\Sigma_w\}_{w\in Y^*}$ be the family of the $\varphi$-deformed quasi-shuffle algebra obtained by duality with $\{\Pi_w\}_{w\in Y^*}$:
\begin{eqnarray*}
\text{for all }u,v\in Y^*,&&\scal{\Sigma_{v}}{\Pi_u}=\delta_{u,v}.
\end{eqnarray*}
\end{enumerate}
\end{definition}
A priori, the $\{\Sigma_w\}_{w\in Y^*}$ 
could be series. We prove first that, in this context, they 
are polynomials.
\begin{proposition}[\sc Adjoint of $\phi_{\pi_1}$]
Let $\phi_{\pi_1}$, be the $\conc$-endomorphism of algebras defined on the letters as follows:
\begin{eqnarray*}
\phi_{\pi_1}:\K\pol{Y}&\longrightarrow&\KY,\\
y_k&\longmapsto&\phi_{\pi_1}(y_k)=\pi_1(y_k).
\end{eqnarray*}
Then  $\phi_{\pi_1}$ is an automorphism with the following properties:
\begin{enumerate}
\item This automorphism is such that, for every $l\in\Lyn Y$, 
$$
\phi_{\pi_1}(P_l)=\Pi_l,
$$ 
where $P_l$ are the polynomials calculated with the mechanism of 
Definition~\ref{Pi}, 
setting $\varphi\equiv 0$ (or, equivalently, by \eqref{pi_l0} with $q=0$), i.e., within the shuffle algebra 
$(\KY,{\tt conc},1_{Y^*},\Delta_{\minishuffle},\epsilon)$.
\item This automorphism has an adjoint $\phi_{\pi_1}^\vee$
  within $\ncp{\K}{Y}$ which reads, on 
the words $w\in Y^*$, 
\begin{equation*}
\phi_{\pi_1}^\vee(w)=\sum_{k\geq 0} \sum_{y_{i_1}\cdots y_{i_k}\in Y}
\scal{w}{\pi_1(y_{i_1})\cdots \pi_1(y_{i_k})}\;y_{i_1}y_{i_2}\cdots y_{i_k}.
\end{equation*}
\item In the style of Definition~\ref{Eulerian}, one has
\begin{align*}
\phi_{\pi_1}&=e+\sum_{k\geq 1} \conc^{(k-1)}\circ(\pi_1\circ I_1)^{\otimes k}\circ\Delta_{\tt conc}^{(k-1)},\\
\phi_{\pi_1}^\vee&=e+\sum_{k\geq 1} \conc^{(k-1)}\circ(I_1\circ \check\pi_1)^{\otimes k}\circ\Delta_{\tt conc}^{(k-1)},
\end{align*}
where $I_1$ is the projector on ${\K}Y$ parallel to
 $\bigoplus_{n\neq 1}(\ncp{\K}{Y})_{n}$.
\item For all $w\in Y^*$, $\Sigma_w=(\phi_{\pi_1}^\vee)^{-1}(S_w)$.
\end{enumerate}
\end{proposition}
\begin{proof}[Proof (sketch)]
It was proved in Theorem~\ref{structure} that the endomorphism
$\phi_{\pi_1}$ is an isomorphism. The recursions used to construct
$\Pi_l$ and $P_l$ prove that $\phi_{\pi_1}(P_l)=\Pi_l$, and then
$\phi_{\pi_w}(P_l)=\Pi_w$ for every word $w$. Now the expression of
$\phi_{\pi_1}$ is a direct consequence of the definition of
$\phi_{\pi_1}$. This implies at once the expression of
$\phi_{\pi_1}^{\vee}$ and the fact that
$\phi_{\pi_1}^\vee\in\mathrm{End}(\ncp{\K}{Y})$. 
The last equality 
comes from the following
\begin{equation*}
\delta_{u,v}=\scal{\Pi_u}{\Sigma_v}=\scal{\phi_{\pi_1}(P_u)}{\Sigma_v}=\scal{P_u}{\phi_{\pi_1}^\vee(\Sigma_v)},
\end{equation*}
which shows that, for all $w\in Y^*$,
$\phi_{\pi_1}^\vee(\Sigma_w)=S_w$
and the claim follows.
\end{proof}
We can now state
the following result.
\begin{theorem}\label{bases}
\begin{enumerate}
\item The family $\{\Pi_l\}_{l\in\Lyn Y}$ forms a basis of $\calP$.
\item The family $\{\Pi_w\}_{w\in Y^*}$ is a linear basis of $\KY$.
\item The family $\{\Sigma_w\}_{w\in Y^*}$ is a linear basis of the $\varphi$-shuffle algebra. 
\item The family $\{\Sigma_l\}_{l\in\Lyn Y}$ forms a pure transcendence basis of $(\KY,\ministuffle_{\varphi},1_{Y^*})$.
\end{enumerate}
\end{theorem}
The first terms of these families, for the $q$-stuffle (see \eqref{pi_l0} and \eqref{recurY-bis}) can be found in \cite{Bui_phd}.

\subsection{Local coordinates by $\varphi$-extended Sch\"utzenberger factorization}
We have observed very early ($\varphi$ needs only to be associative)
that the set of group-like series (for
$\Delta_{\ministuffle_{\varphi}}$) forms a (infinite-dimensional Lie)
group (see Lemmas~\ref{lemme2} and \ref{lemme3}), its Lie algebra is
the (Lie) algebra of Lie series, and we have a nice log-exp
correspondence (see Lemma~\ref{lemme2}). We will see in this 
paragraph
that, when $\varphi$ possesses all the ``good'' properties (moderate,
dualizable, associative and commutative), we have an analogue of the
Wei--Norman theorem~\cite{Bui,weinorman1,weinorman2} which gives a
system of local coordinates for every finite-dimensional (real or
complex) Lie group. Let us recall it here.
\begin{theorem}[\cite{Bui,weinorman1,weinorman2}]
Given a (finite-dimensional) Lie group $G$ (real ${\bf k}=\mathbb{R}$ or complex ${\bf k}=\mathbb{C}$), its Lie algebra $\mathfrak{g}$, and a basis $B=(b_i)_{1\leq i\leq n}$ of $\mathfrak{g}$, there exists a neighbourhood $W$ of\/ $1_G$ (in $G$) and $n$ {\it local coordinate} analytic functions 
$$
W\rightarrow {\bf k},\ (f_i)_{1\leq i\leq n}
$$
such that, for all $g\in W$, we have
\begin{equation*}
g=\prod_{1\leq i\leq n}^{\rightarrow}  e^{t_i(g)b_i}=e^{t_1(g)b_1}e^{t_2(g)b_2}\dots e^{t_n(g)b_n}.
\end{equation*}
\end{theorem}

Now, we have seen that, if $\varphi$ is moderate, dualizable, associative and commutative, 
\begin{equation}
\calH_{{\ministuffle_{\varphi}}}=(\KY,{\tt conc},1_{Y^*},\Delta_{\ministuffle_{\varphi}},\epsilon)
\end{equation}
is isomorphic to the shuffle bialgebra algebra $(\KY,{\tt conc},1_{Y^*},\Delta_{\shuffle},\epsilon)$, 
therefore one can construct bases $\{\Pi_w\}_{w\in Y^*};\ \{\Sigma_w\}_{w\in
  Y^*}$ of $\ncp{\K}{Y}$ with the following properties:
\begin{enumerate}
\item the restricted family $\{\Pi_l\}_{l\in \Lyn Y}$  is a basis of $\calP=\mathrm{Prim}(\ncp{\K}{Y})$;
\item the whole basis is constructed by decreasing concatenation (see Definition~\ref{Pi}) and hence of type PBW;
\item they are in duality $\scal{\Pi_u}{\Sigma_v}=\delta_{{u,v}}$;
\item due to these three properties, we have
\begin{eqnarray}
\Sigma_w=\frac{\Sigma_{l_1}^{\stuffle i_1}\stuffle\cdots\stuffle\Sigma_{l_k}^{\stuffle i_k}}{i_1!\cdots i_k!},&&
\mbox{for }w=l_1^{i_1}\ldots l_k^{i_k}.
\end{eqnarray} 
\end{enumerate}

Now, within the algebra of double series (whose support is $\K^{Y^*\otimes Y^*}$) endowed with the law $\stuffle_{\varphi}\hat{\otimes}\conc$, M.-P. Sch\"utzenberger (see \cite{reutenauer}) gave the beautiful formula 
\begin{equation}\label{Sch_fact}
\sum_{w\in Y^*}w\otimes w=\prod_{l\in \Lyn Y}^{\searrow} e^{\Sigma_l\hat{\otimes} P_l},
\end{equation}
which can be used to provide a system of local coordinates on the {\it
  Hausdorff group}, i.e., the group of series in $\ncs{\K}{Y}$ which
are group-like for $\Delta_{\ministuffle_{\varphi}}$. 
Indeed, due to the fact that for a  group-like $S$, $(S\hat{\otimes}\mathrm{Id})$ is compatible with the law of the double algebra, we get\footnote{All summabilities can be checked from the fact that $\varphi$ is moderate.}, applying the operator $(S\hat{\otimes}\mathrm{Id})$ to (\ref{Sch_fact}),
\begin{equation}
S=(S\hat{\otimes}\mathrm{Id})(\sum_{w\in Y^*}w\hat{\otimes} w)=\prod_{l\in \Lyn Y}^{\searrow} e^{\scal{S}{\Sigma_l}\;P_l},
\end{equation}
which is the perfect analogue of the theorem of Wei and Norman for the Hausdorff group (group of group-like series).  

\section{Conclusion}
In this paper, we have systematically studied the deformations 
of the shuffle product by addition of a superposition term. 
Fortunately, this study provides necessary and sufficient 
conditions for the objects (antipode, Ree ideal, bases in duality) and operators 
(infinite convolutional series, primitive projectors) to exist 
together with their consequences. We have established a local system 
of coordinates for the (infinite-dimensional) Lie group of group-like series. 
This system is the perfect analogue of the well-known theorem of Wei and Norman 
which holds for every finite-dimensional Lie group.

\end{document}